\documentclass[prl,twocolumn, aps,showpacs,superscriptaddress]{revtex4}

\usepackage{latexsym}
\usepackage{amsmath}
\usepackage{amsthm}
\usepackage{amssymb}
\usepackage{amsfonts}

\usepackage{revsymb}

\usepackage{graphicx}
\usepackage[caption=false]{subfig}

\usepackage{graphicx}
\usepackage{tikz}
\usetikzlibrary{decorations.markings,decorations.pathreplacing,calc,fadings,arrows, shadings, patterns}

\newcommand{\be}{\begin{equation}}
\newcommand{\ee}{\end{equation}}
\newcommand{\ba}{\begin{array}}
\newcommand{\ea}{\end{array}}
\newcommand{\bea}{\begin{eqnarray}}
\newcommand{\eea}{\end{eqnarray}}

\newcommand{\ket}[1]{| #1 \rangle}
\newcommand{\bra}[1]{\langle #1 |}

\newtheorem{lemma}{Lemma}

\newtheorem{theorem}{Theorem}

\begin{document}


\title{Universal Adiabatic Quantum Computation via the Space-Time Circuit-to-Hamiltonian Construction}

\author{David \surname{Gosset}}
\affiliation{Institute for Quantum Computing and Dept. of Combinatorics and Optimization, University of Waterloo}
\author{Barbara M. \surname{Terhal}}
\author{Anna \surname{Vershynina}}
\affiliation{JARA Institute for Quantum Information, RWTH Aachen University}

\date{\today}

\begin{abstract}
We show how to perform universal adiabatic quantum computation using a Hamiltonian which describes a set of particles with local interactions on a two-dimensional grid. A single parameter in the Hamiltonian is adiabatically changed as a function of time to simulate the quantum circuit. We bound the eigenvalue gap above the unique groundstate by mapping our model onto the ferromagnetic XXZ chain with kink boundary conditions; the gap of this spin chain was computed exactly by Koma and Nachtergaele using its $q$-deformed version of SU(2) symmetry. We also discuss a related time-independent Hamiltonian which was shown by Janzing to be capable of universal computation. We observe that in the limit of large system size, the time evolution is equivalent to the exactly solvable quantum walk on Young's lattice.

\end{abstract}


\maketitle


Adiabatic quantum computation \cite{FGGS:adia} is a computational model where the groundstate of a simple Hamiltonian is converted into the groundstate of a more complicated Hamiltonian using adiabatic evolution with a slowly changing Hamiltonian. This model was shown to be equivalent to the standard quantum circuit model \cite{ADLLKR:adia} through the use of the Feynman-Kitaev circuit-to-Hamiltonian construction \cite{feynman:qmc, KSV:computation}. Although the class of universal Hamiltonians originally considered (nearest neighbor interactions between six dimensional particles in two dimensions) is not practically viable, perturbation gadget techniques \cite{kkr:hamsiam, OT:qma} were later used to massage it into simpler forms \cite{SV:qma, BL:adia}. However, these techniques have the disadvantage of requiring impractically high variability in the coupling strengths which appear in the Hamiltonian (see, e.g., the analysis in \cite{BGA:resource}). Given this state of affairs, it is of interest to consider how to construct a universal adiabatic quantum computer with a simple Hamiltonian without using perturbative gadgets.

An alternative type of circuit-to-Hamiltonian mapping which is conceptually distinct from the Feynman-Kitaev construction has been used by some authors \cite{margolus, janzing:pra,MMC:groundstate, mizel:ft,MLM:qadiabatic, CGW:BH,BT:spacetime}.  In these works a quantum circuit is mapped to a Hamiltonian which acts on a Hilbert space with computational and ``local'' clock degrees of freedom associated with every qubit in the circuit. This idea was first explored by Margolus in 1989 \cite{margolus}, just four years after Feynman's celebrated paper on Hamiltonian computation \cite{feynman:qmc}. Margolus showed how to simulate a one-dimensional cellular automaton by Schr\"{o}dinger time evolution with a time-independent Hamiltonian. More recently, Janzing \cite{janzing:pra} presented a scheme for universal computation  with a time-independent Hamiltonian. In reference \cite{MLM:qadiabatic} it was claimed that an approach along these lines can be used to perform universal adiabatic quantum computation; unfortunately, the analysis presented by Mizel et al. does not establish the claimed results.  The local clock idea was developed further in the recent ``space-time circuit-to-Hamiltonian construction'' and was used to prove that approximating the ground energy of a certain class of interacting particle systems is QMA-complete \cite{BT:spacetime}. 

Our main result is a new method which achieves efficient universal adiabatic quantum computation using the space-time circuit-to-Hamiltonian construction. The Hamiltonian we use describes a simple system of interacting particles which live on the edges of a two dimensional grid. To prove that the resulting algorithm is efficient we use a mapping from our Hamiltonian to the ferromagnetic XXZ model with kink boundary conditions \cite{KM:XXZchain}.  Our work can be viewed as a carefully tuned adaptation of the proposal from reference \cite{janzing:pra} to the quantum adiabatic setting. In the final part of this work we turn our attention to Janzing's proposal for computation with a time-independent Hamiltonian and we present a new analysis based on the quantum walk on Young's lattice.


\paragraph{Universal adiabatic quantum computation} We consider the universal circuit family used in reference \cite{janzing:pra} and depicted in  Fig.~\ref{fig:circform}(a), i.e., $2n$-qubit circuits which can be schematically drawn as a rotated $n \times n$ grid (shown in Fig.~\ref{fig:circform}(b)) where each plaquette $p$ on the grid corresponds to a two-qubit gate $U_p$. For technical reasons we further restrict the circuit so that many of the gates are fixed to be the identity; in particular, we set $k=\frac{\sqrt{n}}{16}$ and select the rotated $k\times k$ subgrid with its left corner in the center of the original lattice as the ``interaction region''; see Fig.~\ref{fig:circform}(c). In this interaction region the gates $U_p$ are unrestricted, elsewhere they are identity gates. 

We map such a circuit to a Hamiltonian $H(\lambda)$ which depends on a single parameter $\lambda\in [0,1]$.  We will demonstrate that (a) $H(\lambda)$ has a unique groundstate for all $\lambda\in [0,1]$, (b)  the groundstate of $H(0)$ can be efficiently prepared, (c) The output of the quantum circuit is obtained with sufficiently high probability by performing a simple measurement in the groundstate of $H(1)$, and (d) the eigenvalue gap above the ground energy of $H(\lambda)$ is lower bounded as $\frac{1}{{\rm poly}(n)}$ for all $\lambda\in [0,1]$.  These properties allow us to efficiently simulate the given quantum circuit using the quantum adiabatic algorithm with interpolating Hamiltonian $H(\lambda)$.   

We consider a multi-particle Fock space where the particles  live on the edges of the rotated $n\times n$ grid, and each particle has a two dimensional internal degree of freedom that encodes a qubit.  For an edge with midpoint that intersects horizontal and vertical coordinates $(t,w)$ (as shown in  Fig.~\ref{fig:circform}(b), these are unrotated coordinates) we define an operator $a_{t,x}[w]$ which annihilates a particle on that edge with internal state $x\in \{0,1\}$, and a number operator $n_{t,x}[w]=a_{t,x}^{\dagger}[w] a_{t,x}[w]$ which counts the number of particles in this state. $H(\lambda)$ is defined using these operators and, as we will see, it conserves the total number of particles on each horizontal line $w$. We restrict our attention to the sector where there is exactly one particle for each $w\in \{1,\ldots,2n\}$; for the rest of this paper we work in this finite-dimensional Hilbert space. The coordinate $t$ can be viewed as a local time variable (local, since different particles may be located on edges with different values of $t$). For our purposes it is irrelevant whether the particles are fermions, bosons or distinguishable particles, since each particle never strays from its horizontal line of edges. 

For a gate $U_p$ with plaquette $p$ bordered by edges $(t,w),(t+1,w),(t,w+1),(t+1,w+1)$, we define
\begin{eqnarray}
H_{\rm prop}^p= - \sum_{\alpha,\beta, \gamma, \delta} \left(\bra{\beta, \delta} \, U_p \, \ket{\alpha, \gamma} \, a^{\dagger}_{t+1,\beta}[w]  a_{t, \alpha}[w]\right. \nonumber \\
\left.a^{\dagger}_{t+1,\delta}[w+1] a_{t,\gamma}[w+1]\right)+ h.c.,
\nonumber
\label{eq:prop}
\end{eqnarray}
which allows nearest-neighbor particles to hop together. When the particles are both located before (or after) the plaquette, $H_{\rm prop}^p$ can map them onto being both located after (or before) it, while their internal qubit degrees of freedom are changed according to $U_p$ (or $U_p^\dagger$). For each $\lambda\in [0,1]$ we define a positive semidefinite operator
\begin{equation*}
H_{\rm gate}^p(\lambda)= {\bf n}_{t}[w] {\bf n}_{t}[w+1]+{\bf n}_{t+1}[w] {\bf n}_{t+1}[w+1]+\lambda H_{\rm prop}^p,
\end{equation*}
where ${\bf n}_t[w]=n_{t,0}[w]+n_{t,1}[w]$. The Hamiltonian $H(\lambda)$ is built out of these gate operators as well as an operator $H_{\rm string}$ which ensures that the time variables for different particles remain synchronized. Consider a state where the $2n$ occupied edges of the grid form a connected string with endpoints at the top and bottom (e.g., the red string in Fig.~\ref{fig:circform}(b)). Such a string can be represented by $2n$ bits $z=z_1z_2\ldots...z_{2n}$, where $0=/$ represents an edge going down and to the left and $1=\backslash$ is an edge going down and to the right, with total Hamming weight ${\rm wt}(z)=n$. The subspace of the Hilbert space with this property can be identified \footnote{This identification is just a matter of notational convenience and we use it throughout the paper. It is formally defined in the following straightforward way. The string $z$ specifies a set of occupied edges; for each $w=1,\ldots,2n$ let $(t_w,w)$ be the occupied edge on the $w$th horizontal line. Then we identify the state $\left(\Pi_{w=1}^{2n} a^{\dagger}_{t_w,x_w}[w]\right)|\rm{vac}\rangle$, where $|\rm{vac}\rangle$ is the vacuum state with no particles, with $|x\rangle|z\rangle$.} with the space
\begin{equation}
\label{eq:space1}
S_{\rm string}=\text{span}\{|x\rangle |z\rangle : \; x,z\in \{0,1\}^{2n}, {\rm wt}(z)=n\}
\end{equation}
where $z$ describes the string and $x$ represents the $2n$-qubit state encoded in the internal degrees of freedom. It is clear that $S_{\rm string}$ is an invariant subspace for each of the gate operators $H_{\rm gate}^p(\lambda)$--acting with these operators on a state in $S_{\rm string}$ can move the string forward (or backward) and modify the internal state of the particles, but the string remains connected and fixed at the bottom and top of the grid. $H(\lambda)$ will contain a term $H_{\rm string}$ which penalizes particle configurations which do not correspond to connected strings; this will ensure that the groundstate of $H(\lambda)$ is in $S_{\rm string}$. We define $H_{\rm string}=\sum_{v}H_{\rm string}^v$ as a sum of terms for each vertex in the grid, where, if vertex $v$ has four incident edges labeled $(t,w),(t+1,w),(t,w+1),(t+1,w+1)$, we let
\begin{eqnarray}
H_{\rm string}^v={\bf n}_t[w]+{\bf n}_{t+1}[w]+{\bf n}_t[w+1]+{\bf n}_{t+1}[w+1]\nonumber\\
-2\left({\bf n}_t[w]+{\bf n}_{t+1}[w]\right)\left({\bf n}_t[w+1]+{\bf n}_{t+1}[w+1]\right).
\end{eqnarray}
For vertices at the boundaries of the grid which have degree $<4$, this definition is modified so that it only includes operators for the edges which are present. Note that $H_{\rm string}\geq 0$ in the Hilbert space we are working in (the space with exactly one particle per horizontal line), and its nullspace is equal to $S_{\rm string}$. More generally, a particle configuration corresponding to a set of occupied edges which form $L$ string segments which are disconnected from one another has energy $2L-2$, the number of ``loose ends".  In particular, the smallest nonzero eigenvalue of $H_{\rm string}$ is $2$.

\begin{figure}[tb]
\centering
\subfloat[]
{
\begin{tikzpicture}
\tikzset{     position label/.style={        below = 3pt,        text height = 1.5ex,        text depth = 1ex     },    brace/.style={      decoration={brace, mirror},      decorate    } }

\begin{scope}[scale=0.3]
\foreach \x in {1,2,3,4}
{
	
	\foreach \y in {1,2,3,4}
	{
	\draw[fill=gray] (\x+\y,\x-\y) rectangle (\x+\y+1,\x-\y-1);
	}
}
\foreach \y in {-4,...,3}
{
\draw (2,\y)--(9,\y);
\node at (2,-6) {};
}
\end{scope}
\end{tikzpicture}
}
\hspace{0.001cm}
\subfloat[]
{
\begin{tikzpicture}[scale=0.37]
\tikzset{     position label/.style={        below = 3pt,        text height = 1.5ex,        text depth = 1ex     },    brace/.style={      decoration={brace, mirror},      decorate    } }
\begin{scope}[rotate=45]
\draw (0,0)--(0,4)--(4,4)--(4,0)--(0,0);
\draw[color=gray] (0,1)--(4,1);
\draw[color=gray] (0,2)--(4,2);
\draw[color=gray] (0,3)--(4,3);
\draw[color=gray] (1,0)--(1,4);
\draw[color=gray] (2,0)--(2,4);
\draw[color=gray] (3,0)--(3,4);
\draw[thick,color=red] (0,0)--(0,1)--(1,1)--(1,2)--(2,2)--(3,2)--(3,4)--(4,4);
\end{scope}
\node at (-2.465,-1) {$1$};
\node at (-2.465+0.707*7,-1) {$8$};
\node at (0,-1.3){$t$};
\foreach \x in {0,...,7}
{
\draw (-2.465+0.707*\x,-0.5)--(-2.465+0.707*\x,-0.25);
}
\draw (-2.465,-0.5)--(-2.465+0.707*7,-0.5);

\foreach \x in {0,...,7}
{
\draw(3.1,0.3+0.707*\x)--(3.35,0.3+0.707*\x);
}
\draw (3.35,0.3)--(3.35,0.3+0.707*7);
\node at (3.85,0.3+0.707*7) {$1$};
\node at (3.85,0.3) {$8$};
\node at (4.15,2.7) {$w$};
\path [use as bounding box] (-3,-2) rectangle (5,6);
\end{tikzpicture}
}
\hspace{0.0001cm}
\subfloat[]
{
\begin{tikzpicture}[scale=0.18]
\tikzset{     position label/.style={        below = 3pt,        text height = 1.5ex,        text depth = 1ex     },    brace/.style={      decoration={brace, mirror},      decorate    } }
\begin{scope}[rotate=45]
\draw (0,0)--(0,8)--(8,8)--(8,0)--(0,0);
\draw[fill=black] (4,4)--(6,4)--(6,2)--(4,2)--(4,4);
\draw[thick,dotted] (4,4)--(4,8);
\draw[thick,dotted] (4,4)--(0,4);
\node at (-1,4) {$\frac{n}{2}$};
\node at (4,9) {$\frac{n}{2}$};
\end{scope}
\node at (0,-3.4){};
\end{tikzpicture}
}
\caption{A quantum circuit of the form shown in (a) (each gray square is a two qubit gate) is mapped to a Hamiltonian which describes a system of interacting particles that live on the edges of the rotated grid shown in (b). In the groundstate, the edges occupied by particles form a connected string, as illustrated by the red line. (c) Many of the gates are fixed to be the identity; the gates which are unrestricted correspond to plaquettes within a $k\times k$ subgrid, the ``interaction region'', with left corner in the center of the grid.}
\label{fig:circform}
\end{figure}
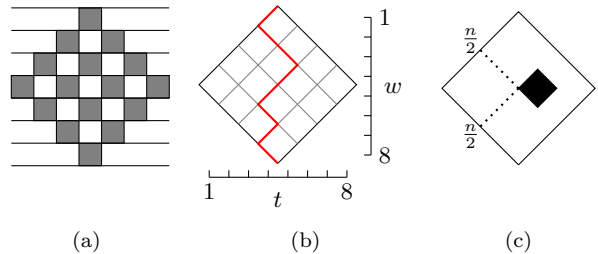
We are now ready to define the Hamiltonian $H(\lambda)$. For $\lambda \in [0,1]$ we let
\begin{align}
H_{\rm circuit}(\lambda)&= \sum_p H_{\rm gate}^p(\lambda)+ \sqrt{1-\lambda^2} H_{\rm init} \nonumber \\
H(\lambda) &=H_{\rm string}+H_{\rm circuit}(\lambda)+H_{\rm input},  \nonumber
\end{align}
where $H_{\rm init}={\bf n}_{n+1}[w=1]+{\bf n}_{n+1}[w=2n]$ is chosen so that in the groundstate of $H(0)$ all particles are located at the left boundary of the grid, and $H_{\rm input}=\sum_{w=1}^{2n} \sum_{t\leq n}  n_{t,1}[w]$ ensures that the internal state of each particle is correctly initialized to $|0\rangle$ when the particle is on the left-hand side of the grid. We now investigate the groundspace of $H(\lambda)$.

To begin, observe that $H_{\rm string}$ commutes with each of the plaquette operators $H^p_{\rm prop}$  \footnote{To see why $[H_{\rm string}, H^p_{\rm prop}]=0$, recall that the eigenvectors of $H_{\rm string}$ with eigenvalue $2L-2$ correspond to particle configurations with $L$ disconnected string segments, and note that applying $H^p_{\rm prop}$ to such a state does not change the number of string segments.} and also with each of the number operators $n_{t,z}[w]$. Thus $[H_{\rm string},H(\lambda)]=0$. As noted above, the ground energy of $H_{\rm string}$ is zero and its first excited energy is $2$. In the following we show that the smallest eigenvalue of $H(\lambda)$ within the space $S_{\rm string}$  is $\sqrt{1-\lambda^2}$. Since  $\sqrt{1-\lambda^2}<2$ this establishes that the corresponding eigenvector of $H(\lambda)$ is the groundstate.

First consider $H(0)$. Since $\sum_p H^p_{\rm gate}(0)$ has minimal energy when the string is either $1^n0^n$ or $z_{\rm init}=0^n1^n$, and since $H_{\rm init}$ penalizes configurations where the first edge of the string is $\backslash$ or the last edge is $/$, we see that the groundspace of $H_{\rm{circuit}}(0)+H_{\rm{string}}$ (with eigenvalue $1$)  is spanned by states $\ket{x}\ket{z_{\rm init}}$. The term $H_{\rm input}$ penalizes all of these states except $\ket{0^{2n}}\ket{z_{\rm init}}$ which is the unique groundstate of $H(0)$, with ground energy $1$. Note that our adiabatic quantum computation can be efficiently initialized since this state is easy to prepare.

To understand the groundspace of $H(\lambda)$ when $\lambda>0$, it will be convenient to work with a different basis for the space $S_{\rm string}$ which builds in the details of the quantum circuit. For any configuration of the string $z\in \{0,1\}^{2n}$ with $\text{wt}(z)=n$, let $V(z)$ be the unitary equal to the product of all the two-qubit gates associated with plaquettes which lie to the left of the string. In other words $V(z)$ is the total unitary of the partially completed circuit with boundary described by $z$. Define basis vectors 
\begin{equation}
|x,z\rangle_{V} =V(z)|x\rangle|z\rangle \qquad x,z\in \{0,1\}^{2n}\;, \text{wt}(z)=n
\label{eq:nicebasis}
\end{equation}
which span $S_{\rm string}$. The action of $H_{\rm{circuit}}(\lambda)$ in this basis has a nice form: it acts nontrivially only on the string degree of freedom; the two-qubit gates which make up the circuit are ``rotated away''.  Moreover, its action on the string register is equivalent (up to a term proportional to the identity and a multiplicative constant) to the ferromagnetic XXZ chain with kink boundary conditions
\begin{eqnarray}
_{V}\langle x',z'|\left(H_{\rm{circuit}}(\lambda)-\sqrt{1-\lambda^2}I\right)|x,z\rangle_{V}\nonumber\\
=2\delta_{x',x} \langle z'|H_{\rm XXZ}(\lambda)|z\rangle
\label{eq:xxz}
\end{eqnarray}
where \cite{KM:XXZchain} (writing $X,Y,Z$ for the Pauli operators) 
\begin{align}
&H_{\rm XXZ}(\lambda)=\frac{1}{4}\sqrt{1-\lambda^2}(Z_{2n}-Z_{1}) \nonumber\\
&-\frac{1}{4}\sum_{w=1}^{2n-1}\left[ (Z_wZ_{w+1}-I)+\lambda(X_wX_{w+1}+ Y_wY_{w+1})\right]\nonumber\\
&=\sum_{w=1}^{2n-1}\ket{\Psi_{q(\lambda)}}\bra{\Psi_{q(\lambda)}}_{w,w+1}, \quad \lambda=\frac{2}{q(\lambda)+q(\lambda)^{-1}},
\label{eq:xxz2}
\end{align}
where $0\leq q(\lambda)\leq 1$ and the $q$-deformed singlet equals $\ket{\Psi_{q}}=\frac{1}{\sqrt{q^2+1}}(\ket{10}-q\ket{01})$. This spin chain can be viewed as a q-analogue of the ferromagnetic Heisenberg chain; it has a remarkable ${\rm SU}_q(2)$ quantum group symmetry which is a deformation of the SU(2) symmetry of the Heisenberg ferromagnet. Its spectral gap, groundspace \cite{KM:XXZchain} and excitations are known \cite{KM:XXZchain, ASW:chain}. In the Supplementary Material we derive an expression for the zero energy groundstate of $H_{\rm XXZ}(\lambda)$ in the sector with Hamming weight $n$. Using this expression and (\ref{eq:xxz}) we immediately obtain a spanning basis for the $\sqrt{1-\lambda^2}$ energy groundspace of $H_{\rm{string}}+H_{\rm{circuit}}(\lambda)$, given by (up to normalization)
\begin{equation}
|\Phi_{\lambda}(x)\rangle=\sum_{z\colon {\rm wt}(z)=n} q(\lambda)^{-A(z)} |x,z\rangle_V \quad x\in \{0,1\}^{2n},
\label{eq:gs_beforeinput}
\end{equation}
where $A(z)=\sum_{j=1}^{2n} jz_j-\frac{n(n+1)}{2}$  is the area of the grid which lies to the right of the string (and $z_j$ is the $j$th bit of $z$). We see that when $\lambda<1$ the associated probability distribution over strings favors the left-hand side of the grid; the most likely string is $z_{\rm init}=0^n1^n$ (with $A(z_{\rm init})=n^{2}$), the least likely is $1^n0^n$ (with $A(z)=0$), etc.  The term $H_{\rm{input}}$ penalizes every state (\ref{eq:gs_beforeinput}) except $|\Phi_{\lambda}(0^{2n})\rangle$, which is the unique groundstate of $H(\lambda)$, with energy $\sqrt{1-\lambda^2}$, for $0<\lambda\leq 1$.

The groundstate $|\Phi_{\lambda=1}(0^{2n})\rangle$ of the final Hamiltonian is a uniform superposition over basis vectors $|0^{2n},z\rangle_V$ corresponding to all possible configurations of the string $z$. To obtain the output of the quantum circuit we measure the locations of the $2k$ particles which lie on horizontal lines that intersect the interaction region. If we find  that all of these particles are located on edges to the right of the interaction region then their internal degrees of freedom give the output of the quantum circuit.  Since the string is connected, this is guaranteed to occur as long as the $n$th particle (i.e., the particle on horizontal line $w=n$) is located on an edge which lies to the right of the interaction region. In the Supplementary Material we show that, with our choice $k=\frac{\sqrt{n}}{16}$, this occurs with probability lower bounded by a positive constant. Finally, we lower bound the eigenvalue gap of $H(\lambda)$.
\begin{theorem}
The smallest nonzero eigenvalue of ${H(\lambda)-\sqrt{1-\lambda^2}I}$ is at least $ \frac{1}{4n+3}(1-\lambda \cos\left(\frac{\pi}{2n})\right)$ for all $\lambda\in [0,1]$. 
\end{theorem}

This $\Omega(n^{-3})$ bound establishes that the adiabatic quantum computation can be performed efficiently. The proof, given in the Supplementary Material, uses the known expression for the eigenvalue gap of $H_{\rm XXZ}(\lambda)$ and a Lemma for bounding the smallest nonzero eigenvalue of an operator sum.

In an attempt to improve the success probability of a final measurement, one might consider modifying this scheme so that the groundstate of the final Hamiltonian is localized at the right side of the grid. This can be achieved by adding another segment to the adiabatic path: after reaching $H(1)$, replace $H_{\rm init}$ with $H_{\rm endit}={\bf n}_{n}[w=1]+{\bf n}_{n}[w=2n]$ and then reduce $\lambda$ from $1$ to $0$. With this choice, every state in the groundspace of the final Hamiltonian has particle configuration corresponding to the string $1^{n} 0^{n}$ on the far right. However the groundspace is degenerate (since $H_{\rm input}|x,1^{n} 0^{n}\rangle_V=0$ for all computational input states $x$). Although the error-free Hamiltonian has a symmetry which prevents transitions between the groundstate corresponding to the correct input and the other wrong-input states, an imperfect realization could potentially derail the computation.


\paragraph{Universal computation with a time-independent Hamiltonian} 
We now discuss a bare-bones version of the related scheme from \cite{janzing:pra}. The quantum circuit family is the same as before, except that now the interaction region is chosen to be the first $K^2$ gates in the circuit with $K=\frac{n}{4}$, i.e., the $K\times K$ subgrid at the far left side of the $n\times n$ grid. The circuit is simulated using Schr\"{o}dinger time evolution with initial state $|0^{2n},z_{\rm init}\rangle_V$ and time-independent Hamiltonian $H_{\rm prop}=\sum_p H_{\rm prop}^p$. After evolving for time $t$, one measures the location of each particle and if one finds them all outside the interaction region then the internal degrees of freedom give the output of the circuit. Janzing's analysis of this scheme is based on an equivalence between $H_{\rm prop}$ and the XY model, which can be diagonalized using a Jordan-Wigner transformation (a unitary mapping to a system of noninteracting fermions in one dimension). In the Supplementary Material we extend one of Janzing's Theorems to prove that the above scheme efficiently simulates a quantum circuit. Specifically we prove that, if the evolution time $t$ is randomly (uniformly) chosen in the interval $[0,T]$ with $T=cn^3$ (for some constant $c$), the probability to measure all the particles outside the interaction region is at least $\frac{1}{4}+\mathcal{O}(\frac{1}{\sqrt{n}})$.

Here we focus on the limit $n\rightarrow\infty$ and directly analyze the time evolution in the given basis without using a Jordan-Wigner transformation. In this way we obtain a detailed picture of the dynamics of the string. To begin, note that a string is associated with a Young diagram (or, equivalently, an integer partition) obtained by rotating the portion of the grid which lies to the left of the string by 45 degrees.  In the limit $n\rightarrow\infty$, the set of string configurations is in one-to-one correspondence with the set of Young diagrams. In the basis (\ref{eq:nicebasis}), $H_{\rm prop}$ acts nontrivially only on the string degree of freedom and it acts on this space as the adjacency matrix $H_\mathbb{Y}$ of Young's lattice, shown in Fig. \ref{fig:younglattice}. In this infinite graph two Young diagrams are connected by an edge if they differ by one box. The dynamics of our system is given by the quantum walk on Young's lattice starting from a very special initial state: the empty partition $\emptyset$. This quantum walk can be solved exactly \cite{S88}; the solution is
\begin{equation}
e^{-iH_{\mathbb{Y}}t}|\emptyset\rangle=e^{-\frac{t^2}{2}}\sum_{m=0}^{\infty}\frac{(-it)^m}{\sqrt{m!}}|\phi_m\rangle,
\label{eq:exactsol}
\end{equation}
where the normalized state $|\phi_m\rangle=\frac{1}{\sqrt{m!}}\sum_{\sigma\dashv m}d_\sigma|\sigma\rangle$, $\sigma\dashv m$ indicates that $\sigma$ is a partition of $m$, and $d_\sigma$ is the dimension of the irreducible representation of the symmetric group $S_m$ associated with $\sigma$ (given by the hook-length formula). For completeness, in the Supplementary Material we review the derivation of equation (\ref{eq:exactsol}).

We see that the quantum walk takes place in a tiny subspace of the full Hilbert space spanned by $\left\{ |\phi_{m}\rangle:\, m\geq0\right\}$. The probability distribution over partitions $\sigma$ as a function of time is given by $p(\sigma,t)=(m!)^{-2}e^{-t^{2}}t^{2m}d_{\sigma}^{2}$ (where $\sigma\dashv m$) which is a \textit{Poissonized Plancherel measure} \cite{R14}. The marginal distribution of $m$  is $\text{Poisson}$ with mean and variance $\mathbb{E}[m]=\text{Var}(m)=t^{2}$. In our case $m$ represents the area to the left of the string (i.e., the number of gates that have been applied) and this shows that, roughly speaking, this area increases quadratically. For large times the random variable $m$ is peaked about its mean in the sense that $\frac{\sqrt{\text{Var}(m)}}{\mathbb{E}[m]}$ is small. The conditional distribution over partitions $\lambda\vdash m$ for fixed $m$ is the widely studied Plancherel measure $\rho_{m}(\sigma)=\frac{d_{\sigma}^{2}}{m!}$, which is known to exhibit a limiting behaviour \cite{L77}. Imagine sampling a partition from $\rho_{m}$, drawing it in the x-y plane and then rescaling both axes by $\frac{1}{\sqrt{m}}$.  As $m\rightarrow\infty$, the resulting picture approaches a fixed shape with probability $\rightarrow1$ \cite{L77,R14} (we include a plot of this shape in the Supplementary Material).  Roughly speaking, for large times we envision the string as a wavefront which moves with constant speed and with scaled shape described by this limit Theorem.

 Finally, note that although it was convenient to consider the limit $n\rightarrow\infty$, we expect this analysis to be approximately valid for finite $n$ when $t$ is small enough so that (\ref{eq:exactsol}) is supported almost entirely on partitions contained in the left-hand side of the rotated $n\times n$ grid.

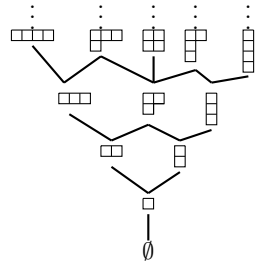
\begin{figure}[t]
\centering
\begin{tikzpicture}[scale=0.14]
\begin{scope}[yshift=1cm]
\begin{scope}[xshift=-1.5 cm]
\draw[color=black] (0,0)--(4,0)--(4,-1)--(0,-1)--(0,0);
\draw[color=black] (1,0)--(1,-1);
\draw[color=black] (2,0)--(2,-1);
\draw[color=black] (3,0)--(3,-1);
\end{scope}

\begin{scope}[xshift=6cm]
\draw[color=black] (0,0)--(3,0)--(3,-1)--(1,-1)--(1,-2)--(0,-2)--(0,0);
\draw[color=black] (1,0)--(1,-1);
\draw[color=black] (2,0)--(2,-1);
\draw[color=black] (0,-1)--(1,-1);
\end{scope}

\begin{scope}[xshift=11cm]
\draw[color=black] (0,0)--(2,0)--(2,-2)--(0,-2)--(0,0);
\draw[color=black] (1,0)--(1,-2);
\draw[color=black] (0,-1)--(2,-1);
\end{scope}

\begin{scope}[xshift=15cm,rotate=90,xscale=-1]
\draw[color=black] (0,0)--(3,0)--(3,-1)--(1,-1)--(1,-2)--(0,-2)--(0,0);
\draw[color=black] (1,0)--(1,-1);
\draw[color=black] (2,0)--(2,-1);
\draw[color=black] (0,-1)--(1,-1);
\end{scope}

\begin{scope}[xshift=21.5cm, rotate=-90]
\draw[color=black] (0,0)--(4,0)--(4,-1)--(0,-1)--(0,0);
\draw[color=black] (1,0)--(1,-1);
\draw[color=black] (2,0)--(2,-1);
\draw[color=black] (3,0)--(3,-1);
\end{scope}
\end{scope}

\begin{scope}[yshift=-5cm,xshift=3cm]
\draw[color=black] (0,0)--(3,0)--(3,-1)--(0,-1)--(0,0);
\draw[color=black] (1,0)--(1,-1);
\draw[color=black] (2,0)--(2,-1);
\end{scope}

\begin{scope}[yshift=-5cm, xshift=11cm]
\draw[color=black] (0,0)--(2,0)--(2,-1)--(1,-1)--(1,-2)--(0,-2)--(0,0);
\draw[color=black] (1,0)--(1,-2);
\draw[color=black] (0,-1)--(2,-1);
\end{scope}

\begin{scope}[yshift=-5cm,xshift=18cm,rotate=270]
\draw[color=black] (0,0)--(3,0)--(3,-1)--(0,-1)--(0,0);
\draw[color=black] (1,0)--(1,-1);
\draw[color=black] (2,0)--(2,-1);
\end{scope}


\begin{scope}[yshift=-10cm,xshift=7cm]
\draw[color=black] (0,0)--(2,0)--(2,-1)--(0,-1)--(0,0);
\draw[color=black] (1,0)--(1,-1);
\end{scope}

\begin{scope}[yshift=-10cm,xshift=15cm,rotate=270]
\draw[color=black] (0,0)--(2,0)--(2,-1)--(0,-1)--(0,0);
\draw[color=black] (1,0)--(1,-1);
\end{scope}


\begin{scope}[yshift=-15cm,xshift=11cm]
\draw[color=black] (0,0)--(1,0)--(1,-1)--(0,-1)--(0,0);
\end{scope}


\node at (11.5,-20){$\emptyset$};


\draw[thick] (11.5,-19)--(11.5, -16.5);
\draw[thick] (11.5,-14.5)--(8, -12);
\draw[thick] (11.5,-14.5)--(14.5, -12.5);
\draw[thick] (14.5,-9.5)--(17.5, -8.5);
\draw[thick] (14.5,-9.5)--(11.5, -8);
\draw[thick] (8,-9.5)--(11.5, -8);
\draw[thick] (8,-9.5)--(4, -7);

\draw[thick] (3.5,-4)--(0.5, -0.5);
\draw[thick] (3.5,-4)--(7, -1.5);
\draw[thick] (12,-4)--(7, -1.5);
\draw[thick] (12,-4)--(12, -1.5);
\draw[thick] (12,-4)--(16, -2.8);
\draw[thick] (17.5,-4)--(16, -2.8);
\draw[thick] (17.5,-4)--(21, -3.4);
\node at (0.5, 3) {$\vdots$};
\node at (7, 3) {$\vdots$};
\node at (12, 3) {$\vdots$};
\node at (16, 3) {$\vdots$};
\node at (21, 3) {$\vdots$};
\end{tikzpicture}
\caption{Young's lattice.}
\label{fig:younglattice}
\end{figure}


\paragraph{Acknowledgements}
We thank Andrew Childs, Robert K\"onig, and Seth Lloyd for helpful discussions. BMT and AV acknowledge funding through the European Union via QALGO FET-Proactive Project No. 600700.  DG is supported in part by NSERC. This research was supported in part by Perimeter Institute for Theoretical Physics. Research at Perimeter Institute is supported by the Government of Canada through Industry Canada and by the Province of Ontario through the Ministry of Economic Development $\&$ Innovation.

\bibliography{refsQMA}


\newpage
\onecolumngrid
\section{Supplementary Material}
\subsection{Mapping to the ferromagnetic XXZ spin chain with kink boundary conditions}
In this Section we present a derivation of equation (\ref{eq:xxz}). In the basis (\ref{eq:space1}) we have, for $z' \neq z$, 
\begin{align}
& \bra{x'}\bra{z'} H_{\rm circuit}(\lambda)\ket{x}\ket{z}=\bra{x'}\bra{z'} \lambda\sum_{p}H^p_{\rm prop}\ket{x}\ket{z}\nonumber\\
&=- \lambda \bra{x'}V (z')V^\dagger(z)\ket{x}\sum_{w=1}^{2n-1}\bra{z'}\left( |01\rangle\langle 10|_{w,w+1} +|10\rangle\langle 01|_{w,w+1}\right)|z\rangle\nonumber\\
&=- \frac{\lambda}{2} \bra{x'}V (z')V^\dagger(z)\ket{x}\sum_{w=1}^{2n-1}\bra{z'}\left( X_{w} X_{w+1}+Y_w Y_{w+1}\right) \ket{z}.
\label{eq:offdiagterms}
\end{align}
Note that this expression is only nonzero when $z'$ differs from $z$ in two consecutive bits $01$ or $10$, and when this happens the unitary $V (z')V^\dagger(z)$ is always either $U_p$ or $U_p^\dagger$ for some plaquette $p$.

Letting $K(z)$ be the number of occurances of $01$ or $10$ in $z$ (the number of ``kinks'' in the string) we have
\begin{align}
& \bra{x'}\bra{z} H_{\rm circuit}(\lambda) \ket{x}\ket{z}\nonumber\\
&=\delta_{x',x}\left(K(z)+\sqrt{1-\lambda^2}(\delta_{z_1,1}+\delta_{z_{2n},0})\right)\nonumber \\
&=\delta_{x',x}\left(\bra{z} \sum_{w=1}^{2n-1} \frac{1}{2}\left(I-Z_w Z_{w+1}\right)+ \frac{1}{2}\sqrt{1-\lambda^2}(Z_{2n}-Z_1)\ket{z} +\sqrt{1-\lambda^2}\right). 
\label{eq:diagterms}
\end{align}

Now let $W=\sum_{z\colon {\rm wt}(z)=n}V(z) \otimes \ket{z}\bra{z}$ so that $W|x\rangle|z\rangle=|x,z\rangle_V$. Using equations  (\ref{eq:offdiagterms}) and (\ref{eq:diagterms}) we see that
\begin{align}
_{V}\langle x',z'|H_{\rm{circuit}}(\lambda)|x,z\rangle_{V} & =\langle x'|\langle z'| W^\dagger H_{\rm{circuit}}(\lambda)W|x\rangle |z\rangle\nonumber \\
&= 2\delta_{x',x} \langle z'|H_{\rm XXZ}(\lambda)|z\rangle+\delta_{z',z}\delta_{x',x}  \sqrt{1-\lambda^2}.
\end{align}


\subsection{Groundspace of $H_{\rm{XXZ}}(\lambda)$} 
In this Section we review the exact solution for the groundspace of $H_{\rm{XXZ}}(\lambda)$ for $0<\lambda\leq 1$. We'll use Bravyi's transfer matrix method for quantum 2-SAT \cite{B06}. Our goal is to show that there is a unique zero energy groundstate in the sector with Hamming weight $n$ and to derive an expression for this state. 

The $q$-deformed singlet is related to the standard singlet state by
\[
(1\otimes T_{q(\lambda)}) |\Psi_{q(\lambda)}\rangle \propto |\Psi_1\rangle
\]
where $\propto$ means proportional to and 
\[
T_q=\left(\begin{array}{cc}
1 & 0\\
0 & q^{-1}
\end{array}\right).
\]
Letting $Q=T_{q(\lambda)}\otimes T_{q(\lambda)}^{2}\otimes T_{q(\lambda)}^{3}\otimes...\otimes T_{q(\lambda)}^{2n}$, we have (for each i)
\[
Q|\Psi_{q(\lambda)}\rangle\langle\Psi_{q(\lambda)}|_{i,i+1}Q=|\Psi_{1}\rangle\langle\Psi_{1}|_{i,i+1}\otimes M_{[2n]\setminus\{i,i+1\}}^{i}
\]
 where $M^{i}$ is a positive operator and we used the fact that $A\otimes A|\Psi_{1}\rangle=\left(\det A\right)|\Psi_{1}\rangle$ for all invertible $A$. From this and equation  (\ref{eq:xxz2}) we see that the nullspace of $QH_{\rm{XXZ}}(\lambda)Q$ is equal to the symmetric subspace (the nullspace of $H_{\rm{XXZ}}(1)$). In other words the zero energy groundspace of $H_{\rm{XXZ}}(\lambda)$ is equal to the image of the symmetric subspace under the invertible mapping Q. In particular, the unique groundstate of $H_{\rm{XXZ}}(\lambda)$ with Hamming weight $n$ is (up to normalization)
\begin{equation}
Q\left(\sum_{\rm{wt}(z)=n}|z\rangle\right)=\sum_{\rm{wt}(z)=n}q(\lambda)^{-f(z)}|z\rangle
\label{eq:chilambda}
\end{equation}
where $f(z)=\sum_{j=1}^{n}jz_{j}$ and $z_j$ is the $j$th bit of $z$. Finally, note that we obtain the same state (up to normalization) if we replace $f(z)$ in the above equation by $A(z)=f(z)-\frac{n(n+1)}{2}$ which is the area of the rotated grid that lies to the right of the string.


\subsection{Lower bound on the success probability for a measurement of the groundstate of $H(1)$}

Recall that we consider a $k\times k$ interaction region rotated and placed with its left corner at the centre of the grid as shown in Figure \ref{fig:circform}(a), with $k=\frac{\sqrt{n}}{16}$. In this Section we consider the particle on the horizontal line labeled $w=n$. We prove that a measurement of the location of this particle in the groundstate $|\Phi_1 (0^{2n})\rangle$ of $H(1)$ results in an edge which lies to the right of the interaction region with probability lower bounded by a constant independent of $n$. If this occurs then the whole string lies to the right of the interaction region and the internal degrees of freedom encode the output of the quantum circuit.  In this way one obtains the $2k$ qubits in the output of the quantum circuit with constant probability.

The state $|\Phi_1 (0^{2n})\rangle$ corresponds to a uniform superposition over configurations of the string. We consider the marginal probability distribution for the edge of the string which lies on the line $w=n$  (since the $w$th edge of the string gives the location of the $w$th particle).

It will be convenient to use the rotated coordinate system shown in Fig. \ref{fig:rotated_grid}. Each vertex of the grid is labeled $(i,j)$ with
$i,j\in\{0,\ldots,n\}$ and each edge is labeled by a tuple $(i,j,x)$ where $(i,j)$ is the location of its ``upper vertex'' (the one closer
to the top of the Figure) and $x\in\{0,1\}$ indicates whether the edge goes left or right out of this vertex. Here $x=0$ corresponds to an edge going left $/$ while $x=1$ is an edge going right $\backslash$. The edges where the $w$th particle (with $w\in\{1,\ldots,2n\}$) may be located are those edges with $i+j=w-1$. For example, when $w\leq n$ these are 
\[
\left\{ (w-1-j,j,x):\, j\in\{0,\ldots w-1\}\: x\in\{0,1\}\right\}.
\]

\begin{figure}
\centering
\begin{tikzpicture}[scale=0.65]
\tikzset{     position label/.style={        below = 3pt,        text height = 1.5ex,        text depth = 1ex     },    brace/.style={      decoration={brace, mirror},      decorate    } }
\begin{scope}[rotate=45]
\draw (0,0)--(0,4)--(4,4)--(4,0)--(0,0);
\draw[color=gray] (0,1)--(4,1);
\draw[color=gray] (0,2)--(4,2);
\draw[color=gray] (0,3)--(4,3);
\draw[color=gray] (1,0)--(1,4);
\draw[color=gray] (2,0)--(2,4);
\draw[color=gray] (3,0)--(3,4);
\node at (-0.25,-0.25){4};
\node at (-0.25,-0.25+1){3};
\node at (-0.25,-0.25+2){2};
\node at (-0.25,-0.25+3){1};
\node at (-0.25,-0.25+4){0};
\draw [decoration={brace},decorate] (-0.75,-0.25) -- (-0.75,3.75);
\node at (-1.25,2){j};
\draw [decoration={brace},decorate] (0.25,4.75) -- (4.25,4.75);
\node at (2,5.25){i};
\node at (4+0.25,4+0.25){0};
\node at (3+0.25,4+0.25){1};
\node at (2+0.25,4+0.25){2};
\node at (1+0.25,4+0.25){3};
\node at (0.25,4+0.25){4};
\end{scope}
\end{tikzpicture}

\caption{\label{fig:rotated_grid} In the rotated coordinate system, vertices are labeled $(i,j)$ with $i,j\in \{0,\ldots,n\}$. The edges which dangle downwards from vertex $(i,j)$ are labeled $(i,j,x)$ with $x=0$ for an edge going left $/$ and $x=1$ for an edge going right $\backslash$.}
\end{figure}
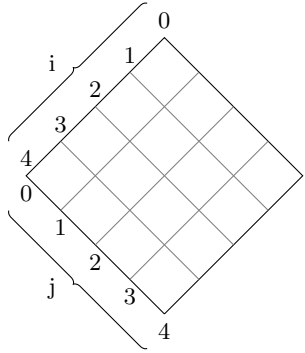

Since the distribution over strings is uniform it is simple to compute the probability $p_{(i,j,x)}$ that the string contains a given edge $(i,j,x)$
\begin{equation}
p_{(i,j,x)}=\frac{\left(\begin{array}{c}
i+j\\
j
\end{array}\right)\left(\begin{array}{c}
2n-(i+j+1)\\
n-j-x
\end{array}\right)}{\left(\begin{array}{c}
2n\\
n
\end{array}\right)}.
\label{eq:pij}
\end{equation}
To see this, note that 
\[
\left(\begin{array}{c}
i+j\\
j
\end{array}\right)
\]
is the number of string segments which connect from the top of the grid to the upper vertex of edge $(i,j,x)$ and 
\[
\left(\begin{array}{c}
2n-(i+j+1)\\
n-j-x
\end{array}\right)
\]
is the number of string segments which connect from the bottom of the grid to the lower vertex of edge $(i,j,x)$. If we fix $w=i+j+1$ and let $j$ and $x$ vary in equation (\ref{eq:pij}) then we get the marginal probability distribution for the location of particle $w$. The $n$th horizontal line intersects the interaction region at $2k=\frac{1}{8}\sqrt{n}$ edges. The total probability for the $n$th particle to be anywhere in the right-hand
side of the grid is equal to $\frac{1}{2}$. Thus, to prove that this particle has constant probability to be found to the right of the
interaction region, it is sufficient to prove a uniform bound $p_{(i,j,x)}\leq\frac{c}{\sqrt{n}}$
for some constant $c<4$ for all edges $(i,j,x)$ with $i+j+1=n$. We now take $n>1$ and prove such a bound with $c=2\sqrt{2}$.

We'll use the fact that 
\[
\frac{4^{m}}{\sqrt{4m}}<\left(\begin{array}{c}
2m\\
m
\end{array}\right)<\frac{4^{m}}{\sqrt{2m}}
\]
(equation (2.15) in reference \cite{K08}).
 Suppose that $n$ is even. Then, for $i+j+1=n$ we get
\begin{equation}
p_{(i,j,x)} <  \frac{\sqrt{4n}}{4^{n}}\left(\begin{array}{c}
i+j\\
j
\end{array}\right)\left(\begin{array}{c}
2n-(i+j+1)\\
n-j-x
\end{array}\right)
< \frac{\sqrt{4n}}{4^{n}}\left(\begin{array}{c}
n\\
\frac{n}{2}
\end{array}\right)\left(\begin{array}{c}
n\\
\frac{n}{2}
\end{array}\right)
 <  \frac{\sqrt{4n}}{4^{n}}\left(\frac{4^{\frac{n}{2}}}{\sqrt{n}}\right)^2
= \frac{2}{\sqrt{n}}\label{eq:p_ij_bnd}
\end{equation}
If $n$ is odd we use almost the same series of inequalities to get

\begin{equation}
p_{(i,j,x)} <  \frac{\sqrt{4n}}{4^{n}}\left(\begin{array}{c}
n-1\\
\frac{n-1}{2}
\end{array}\right)\left(\begin{array}{c}
n+1\\
\frac{n+1}{2}
\end{array}\right)
 <  \frac{\sqrt{4n}}{4^{n}}\left(\frac{4^{\frac{n-1}{2}}}{\sqrt{n-1}}\right)\left(\frac{4^{\frac{n+1}{2}}}{\sqrt{n+1}}\right)
\leq \frac{2}{\sqrt{n-1}}
\leq\frac{2\sqrt{2}}{\sqrt{n}}
\end{equation}

where in the last line we used the fact that $n\geq2$ implies $n-1\geq \frac{n}{2}$. This completes the proof.


\subsection{Proof of Theorem 1}

The following lemma was used implicitly in reference \cite{MMC:groundstate}; here we quote a version of this lemma from \cite{CGW:BH} (proven in Section E.2 of that paper).

 If $M$ is positive semidefinite, write $\gamma(M)$ for its smallest nonzero eigenvalue and let $H|_S$ be the restriction of an operator $H$ to a subspace $S$.

\begin{lemma}[Nullspace Projection Lemma \cite{MMC:groundstate}; this version quoted from \cite{CGW:BH}]
Let $H_A,H_B\geq 0$. Suppose the nullspace $S_A$ of $H_A$ is nonempty and
\[
\gamma(H_B|_{S_A})\geq c>0 \quad \text{and} \quad \gamma(H_A)\geq d>0.
\]
Then 
\[
\gamma(H_A+H_B)\geq \frac{cd}{c+d+\|H_B\|}.
\]
\label{lem:npl}
\end{lemma}

We'll also use the following Theorem, proven in reference \cite{KM:XXZchain}.
\begin{theorem}[\cite{KM:XXZchain}]
The smallest nonzero eigenvalue of the $2n$-qubit XXZ chain with kink boundary conditions (\ref{eq:xxz2}) is
\[
\gamma \left(H_{\rm{XXZ}}(\lambda)\right)=1-\lambda \cos(\frac{\pi}{2n}).
\]
\label{thm:xxzthm}
\end{theorem}
\begin{figure}[htb]
\centering
\begin{tikzpicture}[scale=0.75]
\tikzset{     position label/.style={        below = 3pt,        text height = 1.5ex,        text depth = 1ex     },    brace/.style={      decoration={brace, mirror},      decorate    } }
\begin{scope}[rotate=45]
\draw (0,0)--(0,4)--(4,4)--(4,0)--(0,0);
\draw[color=gray] (0,1)--(4,1);
\draw[color=gray] (0,2)--(4,2);
\draw[color=gray] (0,3)--(4,3);
\draw[color=gray] (1,0)--(1,4);
\draw[color=gray] (2,0)--(2,4);
\draw[color=gray] (3,0)--(3,4);
\draw[thick,color=red] (0,0)--(0,2)--(2,2)--(3,2)--(3,3)--(4,3)--(4,4);
\draw[thick, dotted] (-0.5,-0.5)--(4.5,4.5);
\node at (4.5,4) {L};
\node at (0.5,2.2) {$e$};
\fill (0,0) circle [radius=4pt];
\fill (2,2) circle [radius=4pt];
\end{scope}

\path[draw=black,solid,line width=1mm,fill=black, preaction={-triangle 90,thick,draw,shorten >=-1mm} ] (3.5,2.8)--(6.5,2.8);
\begin{scope}[xshift=10cm]
\begin{scope}[rotate=45]
\draw (0,0)--(0,4)--(4,4)--(4,0)--(0,0);
\draw[color=gray] (0,1)--(4,1);
\draw[color=gray] (0,2)--(4,2);
\draw[color=gray] (0,3)--(4,3);
\draw[color=gray] (1,0)--(1,4);
\draw[color=gray] (2,0)--(2,4);
\draw[color=gray] (3,0)--(3,4);
\draw[thick,color=red] (0,0)--(2,0)--(2,2)--(3,2)--(3,3)--(4,3)--(4,4);
\draw[thick,dotted] (-0.5,-0.5)--(4.5,4.5);
\node at (4.5,4) {L};
\node at (2.3,0.5) {$e'$};
\fill (0,0) circle [radius=4pt];
\fill (2,2) circle [radius=4pt];
\end{scope}
\end{scope}
\end{tikzpicture}
\caption{The bijection used in Lemma \ref{lem:pr} maps each string passing through an edge $e$ to a string passing through the mirror image edge $e'$. In this example the portion of the string which lies between the two vertices indicated by black circles is reflected about the vertical line $L$.}
\label{fig:bijection}
\end{figure}
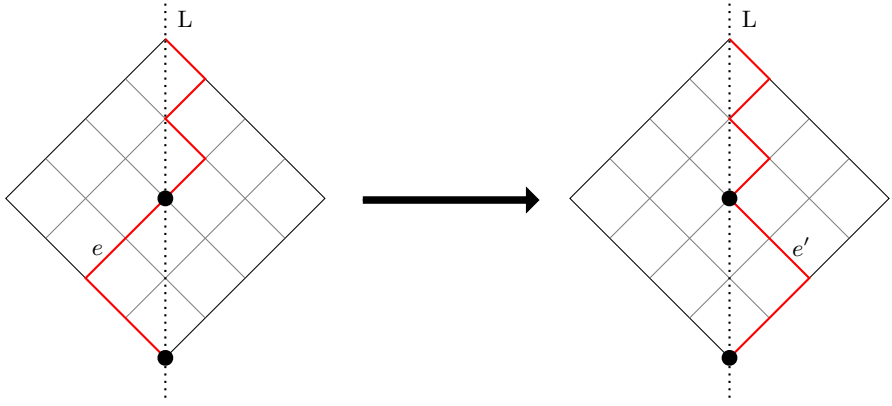

\noindent \textit{Proof of Theorem 1}\\

We first recall some facts established in the main text of the paper. We work in the Hilbert space with exactly one particle per horizontal line $w\in \{1,\ldots,2n\}$. Restricted to this space,  $H_{\rm{string}}$ has ground energy $0$ and smallest nonzero eigenvalue $\gamma (H_{\rm{string}})=2$. Furthermore, $[H_{\rm{string}},H(\lambda)]=[H_{\rm{string}},H_{\rm circuit} (\lambda)]=0$, and the unique groundstate of $H(\lambda)$ has energy $\sqrt{1-\lambda^2}$ and lies in the nullspace $S_{\rm string}$ of $H_{\rm{string}}$.

Now let $\bar{H}(\lambda)=H(\lambda)-\sqrt{1-\lambda^2}I\geq 0$, and write
\[
\bar{H}(\lambda)=H_A+H_B
\]
as a sum of two positive semidefinite terms $H_A=H_{\rm string}+(H_{\rm circuit}(\lambda)-\sqrt{1-\lambda^2}I)$ and $H_B=H_{\rm input}$. We use the Nullspace Projection Lemma to bound the smallest nonzero eigenvalue of $\bar{H}(\lambda)$. To do so, we'll need bounds on $\gamma(H_A)$, $\gamma(H_B|_{S_{A}})$ and $\|H_B\|$, where $S_{A}$ is the nullspace of $H_A$. 

First consider $\gamma (H_A)$. Since $[H_{\rm{string}},H_{\rm circuit} (\lambda)]=0$, $H_{\rm circuit} (\lambda)\geq 0$ and $\gamma (H_{\rm{string}})=2$ we see that any eigenvector of $H_{A}$ which is orthogonal to $S_{\rm string}$ has eigenvalue at least $2-\sqrt{1-\lambda^2}$. So $\gamma (H_A)$ is lower bounded by either $2-\sqrt{1-\lambda^2}$ or 
\[
\gamma \left((H_{\rm circuit}(\lambda)-\sqrt{1-\lambda^2}I)|_{S_{\rm string}}\right)
\]
whichever is smaller. By equation (\ref{eq:xxz}) and Theorem \ref{thm:xxzthm} we get 
\[
\gamma \left((H_{\rm circuit}(\lambda)-\sqrt{1-\lambda^2})|_{S_{\rm string}}\right)=\gamma \left(2H_{\rm XXZ}(\lambda)|_{W_n}\right)\geq 2\gamma \left(H_{\rm XXZ}(\lambda)\right)=2-2\lambda \cos(\frac{\pi}{2n})
\]
where $W_n=\text{span} \{|z\rangle : z\in \{0,1\}^{2n},\text{wt}(z)=n\}$ is the Hamming weight $n$ subspace. To see why the third inequality holds, note that $H_{\rm XXZ}(\lambda)$ conserves Hamming weight so its smallest nonzero eigenvalue in the sector with Hamming weight $n$ is at least the smallest nonzero eigenvalue in the full Hilbert space. Putting all of this together we get
\begin{equation}
\gamma (H_A)\geq \min\{2-2\lambda \cos(\frac{\pi}{2n}),2-\sqrt{1-\lambda^2}\}\geq 1-\lambda \cos(\frac{\pi}{2n}).
\label{eq:gammaA}
\end{equation}

Next we bound $\gamma(H_B|_{S_{A}})$. The nullspace of $H_A$ is spanned by the states $|\Phi_\lambda (x)\rangle$ with $x\in \{0,1\}^{2n}$ which are defined in equation (\ref{eq:gs_beforeinput}) for $0<\lambda\leq 1$. For $\lambda=0$ we set $|\Phi_0 (x)\rangle=|x,z_{\rm init}\rangle_V$.  In this basis $H_B|_{S_{A}}$ is diagonal, with matrix elements
\begin{equation}
 \frac{\langle \Phi_\lambda (x)|H_B| \Phi_\lambda (y)\rangle}{\| |\Phi_\lambda (x)\rangle \| \||\Phi_\lambda (y)\rangle\|}=\delta_{x,y}\sum_{w=1}^{2n} p_w(\lambda)\delta_{x_w,1}
\label{eq:HBrestricted}
\end{equation}
where $p_w(\lambda)$ is the probability of finding particle $w$ on the left half of the grid in the state $|\Phi_\lambda(x)\rangle$ (note this probability does not depend on $x$). When $\lambda=1$ we have $p_w(1)=\frac{1}{2}$ since $|\Phi_1(x)\rangle$ is a uniform superposition over the states $|x,z\rangle_V$ where $z$ runs over all configurations of the string. Moreover, as can be seen from equation (\ref{eq:gs_beforeinput}), when $0<\lambda<1$ the probability of a string is proportional to
\[
q(\lambda)^{-2A(z)}
\]
where $A(z)$ is the area of the rotated grid which lies to the right of the string $z$. Since this is an increasing function of $A(z)$, one might intuitively expect that the probability for a given particle to be found on the left-hand side of the grid does not decrease below $\frac{1}{2}$ as $\lambda$ goes below $1$. Indeed, we prove that $p_w(\lambda)$ is lower bounded by $\frac{1}{2}$ for $0\leq\lambda\leq1$.

\begin{lemma}
For each $1\leq w\leq 2n$ and $0 \leq\lambda\leq 1$, $p_w (\lambda)\geq \frac{1}{2}$.
\label{lem:pr}
\end{lemma}
\begin{proof}
When $\lambda=0$ the string $z_{\rm init}$ has probability $1$ and the result clearly holds; we therefore consider $0<\lambda\leq1$. 

To prove the Lemma, it is sufficient to show that the total probability that the string passes through some edge $e=(t,w)$ on the left
hand side (i.e., with $t\leq n$) is greater than or equal to the total probability that the string passes through the mirror image edge $e^{\prime}=(2n+1-t,w)$ on the right-hand side, obtained by reflecting about the vertical line $L$ which passes through the centre of the grid.

 Let $\Delta^{c}$ be the set of all strings which pass through
a given edge $c$. To prove the result, we define a bijection which
maps each string $z\in \Delta^{e}$ to $z^{\prime}\in \Delta^{e^{\prime}}$
such that the probability of $z$ is at least that of $z^{\prime}$. The bijection we use is illustrated in Fig. \ref{fig:bijection}. For each $z\in \Delta^{e}$, consider the vertices on the string which are located on the vertical line $L$ (the first and last vertex of
the string are always on $L$, and in general there may be more).
In particular, consider the two vertices $v_{\text{1}},v_{2}\in L$
on the string immediately before and after the string visits the edge
$e$ (these vertices are indicated in black in the Figure). Let $z^{\prime}$ be the string obtained from $z$ by reflecting
the segment of $z$ between $v_{1},v_{2}$ about the line $L$. The
mapping $z\rightarrow z^{\prime}$ is clearly a bijection. Furthermore,
the string $z^{\prime}$ passes through the edge $e^{\prime}$ and
satisfies $A(z^{\prime})\leq A(z)$. Hence the probabilty of $z$ is greater than that of $z^{\prime}$, for all $0<\lambda\leq1$ .

\end{proof}

Using equation (\ref{eq:HBrestricted}) and Lemma \ref{lem:pr} we arrive at the lower bound 
\begin{equation}
\gamma \left(H_B|_{S_A}\right) = \min_{w} p_w(\lambda) \geq \frac{1}{2}.
\label{eq:gammaB}
\end{equation}. 

Finally since our Hilbert space contains $2n$ particles in total we have
\begin{equation}
\|H_B\|=\left\|\sum_{w=1}^{2n} \sum_{t\leq n}  n_{t,1}[w]\right\|\leq 2n.
\label{eq:Hbnorm}
\end{equation}

Now plugging the bounds (\ref{eq:gammaA}), (\ref{eq:gammaB}), and (\ref{eq:Hbnorm}) into the Nullspace Projection Lemma we get
\begin{align}
\gamma (\bar{H}(\lambda))& \geq \frac{\frac{1}{2}(1-\lambda \cos(\frac{\pi}{2n}))}{\frac{1}{2}+(1-\lambda \cos(\frac{\pi}{2n}))+2n}\nonumber\\
& \geq \frac{1}{4n+3}(1-\lambda \cos(\frac{\pi}{2n})).
\end{align}
\qed


\subsection{Torus geometry}
In \cite{BT:spacetime} the space-time construction was analyzed for a torus (with periodic boundary conditions in space and time) instead of a rotated $n \times n$ grid.   In this geometry, the string around the torus is not held fixed at boundary points but can freely move.  Here we discuss the Hamiltonian $H_{\rm prop}=\sum_{p}H_{\rm prop}^p$ in this geometry and we show how its action on the string degree of freedom is equivalent to a lattice model of a persistent current ring.

In the torus geometry, the state-space has basis vectors $\ket{x} \ket{\tau}\ket{z}$ where $x$ is the internal $2n$ qubit state, $\tau \in \mathbb{Z}_D$  is a boundary point and $z \in \{0,1\}^{2n}$  describes a  string which forms a closed loop. The Hamiltonian $H_{\rm prop}$ conserves $\text{wt}(z)$ and, if $n$ is sufficiently large then there exist closed loops with $\text{wt}(z)\neq n$ and nonzero winding number. It is only possible to have nonzero winding number when  $2n > D$. For computational purposes we are interested in the sector with $\text{wt}(z)=n$ and winding number $0$, which we now specialize to. It was shown in \cite{BT:spacetime} that the gates can be rotated away and the action of $H_{\rm prop}$ on the string degree of freedom is unitarily equivalent to 
\begin{equation}
H_{\rm XY}^{\partial}=-\frac{1}{2}\sum_{w=1}^{2n-1} (X_w X_{w+1}+Y_w Y_{w+1})-\sum_{\tau \in \mathbb{Z}_{D}} \left(\ket{\tau-1}\bra{\tau} \otimes \sigma_{w=1}^- \sigma_{w=2n}^++\ket{\tau} \bra{\tau-1} \otimes \sigma_{w=1}^+ \sigma_{w=2n}^-\right). \nonumber
\end{equation}
The Hamiltonian is block diagonal in a basis where the boundary variable $\tau$ is in a plane-wave state, i.e., 
\begin{equation}
\ket{\psi_k}=\frac{1}{\sqrt{D}}\sum_{\tau \in \mathbb{Z}_{D}}e^{2\pi i k \tau/D}\ket{\tau}, \text{  }k=0,...,D-1, \;\;H_{\rm XY}^{\partial}\ket{\psi_k}\ket{z}=\ket{\psi_k} H_{\rm XY}^{\partial}(k)\ket{z}, \nonumber
\end{equation}
where
\begin{equation*}
H_{\rm XY}^{\partial}(k)= -\frac{1}{2}\sum_{w=1}^{2n-1}(X_w X_{w+1}+Y_w Y_{w+1})- e^{2\pi ik/D}\sigma_1^-\sigma_{2n}^+- e^{-2\pi ik/D}\sigma_1^+\sigma_{2n}^-.
\end{equation*}
For even $n$, $H_{\rm XY}^{\partial}(k)$ is unitarily equivalent (via a Jordan-Wigner transformation to fermion operators $Z_1\otimes Z_2\ldots \otimes Z_{w-1}\otimes \sigma_w^- \rightarrow a_w$ and a phaseshift $a_w \rightarrow e^{-2 \pi i k w/(2nD)} a_w$) to a well-known lattice model of a `persistent current ring' 
\begin{equation}
H_{\phi(k)}=-\sum_{w=1}^{2n} (e^{-2 \pi i \phi(k)/2n} a_w^{\dagger} a_{w+1}+ e^{2 \pi i \phi(k)/2n} a_{w+1}^{\dagger} a_{w}) \nonumber
\end{equation}
with $k$-dependent magnetic flux-variable $\phi(k)=\frac{k}{D}$ going through the ring and periodic boundary conditions $a_{2n+1}=a_1$.  Note that if one considers using the torus geometry to simulate a quantum circuit, one cannot initialize the computation in a state where the boundary register is $\ket{\psi_k}$ (which is delocalized over the whole lattice). It might however be interesting to consider whether such a simulation is possible using Schr\"{o}dinger time evolution where the initial state of the boundary register is a wavepacket.


\subsection{Efficiency of the scheme for universal quantum computation with a time-independent Hamiltonian}

Here we extend Janzing's analysis of the scheme for universal computation with a time-independent Hamiltonian. Janzing proves two relevant Theorems \cite{janzing:pra}; however, neither provides a polynomial bound on the total resource requirements. The statement of his first Theorem (labeled Theorem 2 in his paper) includes a condition $n\gg K$, and the precise meaning of this condition is not clear. The second Theorem (Theorem 3) concerns finite system size but proves a limiting result for $T\rightarrow \infty$ and it is not immediately clear whether $T=\text{poly}(n)$ suffices. Our goal in this section is to prove that the scheme is efficient.

Recall that the setup is as follows. We choose the interaction region as a subgrid with side length $K=n/4$, positioned in the left corner of the rotated $n \times n$ grid, and we consider the time evolved state
\[
|\chi_t\rangle=e^{-iH_{\rm prop}t}\ket{0^{2n},  z_{\rm init}}_V
\]
where $z_{\rm init}=0^n1^n$, $H_{\rm prop}=\sum_p H_{\rm prop}^p$, and with time $t$ uniformly chosen in an interval $[0,T]$ where $T=\Theta(n^3)$. After the evolution one measures the location of each of the particles. If each of the particles is found to the right of the interaction region then their internal degrees of freedom give the output of the computation. The following result is a finite $T$ variant of Theorem 3 in \cite{janzing:pra} and provides an $\Omega(1)$ lower bound on the probability of finding the whole string outside the interaction region. This establishes that this scheme can be used to efficiently simulate a universal quantum computer. The proof modifies Janzing's Theorem 3 \cite{janzing:pra} following a strategy from reference \cite{expqwalk}.

\begin{theorem}
Let $K=\frac{n}{4}$ and let the interaction region be the $K\times K$ subgrid positioned in the left corner of the rotated $n\times n$ grid. There exists a positive constant $c$ such that the following holds.  Choose $t \in [0,T]$ uniformly at random, with $T=c n^{3}$. Then the probability of measuring each of the particles outside the interaction region in the state $|\chi_t\rangle$  is at least $\frac{1}{4}+\mathcal{O}(\frac{1}{\sqrt{n}})$.
\end{theorem}

\begin{proof}
We use the fact that the action of $H_{\rm prop}$ in the subspace $S_{\rm string}$ is equivalent to the XY model 
\begin{equation}
_{V}\langle x',z'|H_{\rm{prop}}|x,z\rangle_{V}=\delta_{x',x} \langle z'|H_{\rm XY}|z\rangle
\label{eq:xy}
\end{equation}
\begin{equation}
H_{\rm XY}=-\frac{1}{2}\sum_{w=1}^{2n-1} \left (X_wX_{w+1}+Y_wY_{w+1}\right).
\label{eq:xy2}
\end{equation}

Using this equivalence we see that the string register of the state $|\chi_t\rangle$ is described by
\[
\ket{\psi_t}=e^{-iH_{\rm XY}t}\ket{z_{\rm init}}.
\]

$H_{\rm XY}$ conserves Hamming weight and this state lives in the sector with Hamming weight $n$. The approach in \cite{janzing:pra} is to relate the dynamics of $H_{\rm XY}$ within the Hamming weight $n$ sector to  properties of the dynamics within the Hamming weight $1$ sector.  To this end, we will need the eigenvalues and eigenvectors of $H_{\rm XY}$ in the Hamming weight $1$ sector. Within this sector $H_{\rm XY}$ is equal to minus the adjacency matrix of a path with $2n$ vertices, with eigenvalues and eigenvectors given by

\begin{equation}\label{eigenvalues}
\lambda_{r}=-2\cos\frac{\pi r}{2n+1} \qquad  \ket{e_r}=\sqrt{\frac{2}{2n+1}}\,\sum_{j=1}^{2n}\sin\Bigl(\frac{jr\pi}{2n+1}\Bigr)\ket{\widehat{j}}
\end{equation}
for $r=1,\ldots, 2n$, where $\widehat{j}=0^{j-1}10^{n-j-1}$ is the $j$th Hamming weight $1$ bit-string.

Define probability distributions 
\[
{\rm Prob}(z,t)=|\bra{z} \psi_t \rangle|^2 \quad \text{and} \quad {\rm Prob}_T(z)=\frac{1}{T}\int_0^T dt \;{\rm Prob}(z,t),
\]
and let $N=N(z)=\sum_{j=1}^n z_j$ be the random variable which counts the Hamming weight of the first $n$ bits of $z$. Let ${\rm Exp}_T[N]=\sum_z {\rm Prob}_T(z) N(z)$. At time $t=0$ the distribution $\rm{Prob}(z,0)$ is concentrated on the string $z_{\rm init}$ which satisfies $N(z_{\rm init})=0$. On the other hand any string $z$ satisfying $N(z) \geq n/4$ lies to the right of the interaction region. Equivalently, $M(z)=n-N(z)\leq 3n/4$ guarantees that $z$ lies to the right of the interaction region. Our goal is to prove that, in the distribution ${\rm Prob}_T$, the probability that $M$ is at least $3n/4$ is asympotically upper bounded by $\frac{3}{4}+\mathcal{O}(\frac{1}{\sqrt{n}})$. Markov's inequality relates the probability of this event with the expected value of $M$:
\begin{equation}
\label{eq:Markov}
{\rm Prob}_T\left(M \geq \frac{3n}{4}\right) \leq \frac{4}{3n}{\rm Exp}_T[M]
\end{equation}
(with a slight abuse of notation on the LHS). To complete the proof we will establish that we can choose $T=c n^3$ (where $c$ is a constant) such that
\begin{equation}
\left|{\rm Exp}_T[N]-\frac{n}{2}\right|\leq \frac{n}{16}+\mathcal{O}(\sqrt{n})
\label{eq:EN}
\end{equation}
and hence ${\rm Exp}_T[M]\leq \frac{n}{2}+\frac{n}{16}+\mathcal{O}(\sqrt{n})$. Plugging this bound into (\ref{eq:Markov}) gives the desired result.

To show (\ref{eq:EN}) we rely on some facts proven in Appendix D of \cite{janzing:pra}. The first is that 
\[
|{\rm Exp}_{\infty}[N]-\frac{n}{2}| =\mathcal{O}(\sqrt{n}).
\]
so it is sufficient for us to prove
\begin{equation}
|{\rm Exp}_{\infty}[N]-{\rm Exp}_{T}[N]| \leq \frac{n}{16}.
\label{eq:ENinf}
\end{equation}
We also use the following relationship between the Hamming weight $n$ sector and the Hamming weight 1 sector, which was established in reference \cite{janzing:pra} (see equation (34) of the arxiv version of that paper) using a Jordan-Wigner transformation
\begin{equation}
{\rm Exp}_T[N]=\sum_{j=1}^n \sum_{l=n+1}^{2n} {\rm Prob}_T(l \rightarrow j)
\label{eq:ET}
\end{equation}
where 
\[
{\rm Prob}_T(l \rightarrow j)=\frac{1}{T}\int_{t=0}^{T}\left|\langle \widehat{j} |e^{-iH_{\rm XY}t}|\widehat{l}\rangle\right|^2 dt
\]
 is a time-averaged transition probability in the Hamming weight 1 sector. Using (\ref{eq:ET}) and expanding each of the transition probabilities in the basis (\ref{eigenvalues}) we get
\begin{align}
|{\rm Exp}_{\infty}[N]-{\rm Exp}_T[N]|& \leq\sum_{l=n+1}^{2n}\sum_{j=1}^n\sum_{1\leq r\neq r'\leq 2n}\Bigl|\langle \widehat{j}|e_r\rangle\langle e_r|\widehat{l}\rangle\langle \widehat{l}|e_{r'}\rangle\langle e_{r'}|\widehat{j}\rangle\Bigr|\frac{|1-e^{-i(\lambda_r-\lambda_{r'})T}|}{T|\lambda_r-\lambda_{r'}|}\nonumber \\
&\leq \frac{2}{T\Delta}\sum_{l=n+1}^{2n}\sum_{j=1}^n \Bigl(\sum_{r=1}^{2n}\left|\langle \widehat{j}|e_r\rangle \langle e_r|\widehat{l}\rangle\right| \Bigr)^2,
\end{align}
where $\Delta=\min_{1\leq r\neq r'\leq 2n}|\lambda_r-\lambda_{r'}|$ is the minimal spectrum gap, and in the first line we used the fact that the terms with $r=r'$ in the sum cancel with ${\rm Exp}_{\infty}[N]$. Using the expression (\ref{eigenvalues}) for the eigenvalues and eigenvectors we have $|\langle \widehat{i}|e_r\rangle|\leq \sqrt{\frac{2}{2n+1}}$ for each $i$ and $\Delta=\Theta(\frac{1}{n^2})$ (by a standard trigonometric identity). Thus, choosing $T=c n^3$ for a sufficiently large constant $c$, we can ensure that (\ref{eq:ENinf}) holds, completing the proof.
\end{proof}

Note that we expect that the $\Theta(n^3)$ choice of evolution time given in the Theorem is not best possible. For example, the analysis given in the main text of the paper suggests that the string moves with constant speed and an evolution time which is linear in $n$ should suffice.


\subsection{Quantum walk on Young's lattice}
In this Section we review the exact solution for the quantum walk on Young's lattice starting from the empty partition \cite{S88}.

Recall that we write $\sigma \dashv k$ to indicate that $\sigma $ is a partition of $k$. Write $\kappa \lessdot\sigma$
to indicate that $\sigma$ covers $\kappa$ in Young's lattice, i.e.,
$\kappa$ and $\sigma$ are partitions of consecutive integers connected
by an edge, with $\kappa \dashv k$ and $\sigma\dashv\left(k+1\right)$ for some $k\geq0$. Let
\begin{align*}
D_{\sigma}^{\downarrow} & =\left|\left\{ \kappa:\,\sigma \gtrdot\kappa\right\} \right|\quad D_{\sigma}^{\uparrow}=\left|\left\{ \kappa:\,\sigma \lessdot\kappa\right\} \right|
\end{align*}
 be the ``down'' and ``up'' degree of a partition $\sigma$,
and let 
\[
B_{\sigma,\sigma'}=\left|\left\{ \kappa:\,\kappa\gtrdot\sigma\text{ and }\kappa\gtrdot\sigma' \right\} \right|\quad C_{\sigma,\sigma'}=\left|\left\{ \kappa:\,\sigma\gtrdot\kappa\text{ and }\sigma'\gtrdot\kappa\right\} \right|
\]
 be the number of elements covering and covered by a pair of partitions
$\sigma,\sigma'$.

Young's lattice has the following two properties

\begin{align}
D_{\sigma}^{\uparrow} & =D_{\sigma}^{\downarrow}+1\qquad B_{\sigma,\sigma'}=C_{\sigma,\sigma'}\label{eq:diffposet}
\end{align}
which allow us to exactly solve for the time evolution of the quantum
walk (these are two of the defining features of a differential poset
\cite{S88}; see also Chapter 5 of \cite{S01}).

Let $A^{\dagger}$ be defined by 
\[
A^{\dagger}|\sigma\rangle=\sum_{\kappa\gtrdot\sigma}|\kappa\rangle,
\]
and let $A$ be its Hermitian conjugate; then
\begin{align*}
AA^{\dagger}|\sigma \rangle & =\sum_{\substack{\kappa\neq\sigma}
}B_{\kappa,\sigma}|\kappa\rangle+D_{\sigma}^{\uparrow}|\sigma\rangle\qquad A^{\dagger}A|\sigma \rangle=\sum_{\substack{\kappa \neq\sigma}
}C_{\kappa,\sigma}|\kappa \rangle+D_{\sigma }^{\downarrow}|\sigma\rangle.
\end{align*}
 Now using equation (\ref{eq:diffposet}) we see that $[A,A^{\dagger}]|\sigma\rangle=|\sigma\rangle$
which shows that $A$ and $A^{\dagger}$ satisfy the canonical commutation relation $[A,A^{\dagger}]=1$. Because of this fact it is easy to exponentiate the adjacency matrix $H_\mathbb{Y}=A+A^\dagger$ of Young's lattice: 
\[
e^{-i(A+A^{\dagger})t}=e^{-\frac{t^{2}}{2}}e^{-iA^{\dagger}t}e^{-iAt}.
\]
Since the initial state $|\emptyset\rangle$ is annihilated by $A$ we get 
\begin{align*}
e^{-i(A+A^{\dagger})t}|\emptyset\rangle & =e^{-\frac{t^{2}}{2}}e^{-iA^{\dagger}t}|\emptyset\rangle\\
 & =e^{-\frac{t^{2}}{2}}\sum_{m=0}^{\infty}\frac{\left(-it\right)^{m}}{m!}A^{\dagger m}|\emptyset\rangle.
\end{align*}
Recall that partitions $\sigma \vdash m$ are in one-to-one correspondence with irreducible representations of the symmetric group $S_{m}$.
The dimension $d_{\sigma}$ of the irrep associated with a given Young diagram $\sigma$ is equal to the number of standard Young
Tableaux with shape $\sigma$ (a standard Young Tableau is a Young diagram where the $m$ boxes are filled with the numbers $\{1,\ldots,m\}$ in such a way that the numbers increase to the right along each row and downward
along each column). Using this fact it is not hard to see that $d_{\sigma}$ is also equal to the number of paths of length $m$ in Young's lattice which start at $\emptyset$ and end at $\sigma$, i.e., 
\[
\langle\sigma |A^{\dagger m}|\emptyset\rangle=d_{\sigma}.
\]
 Hence 
\begin{align*}
e^{-i(A+A^{\dagger})t}|\emptyset\rangle & =e^{-\frac{t^{2}}{2}}\sum_{m=0}^{\infty}\frac{\left(-it\right)^{m}}{\sqrt{m!}}|\phi_{m}\rangle
\end{align*}
 where, for each $m\geq0$,
\[
|\phi_{m}\rangle=\sum_{\sigma \vdash m}\frac{d_{\sigma}}{\sqrt{m!}}|\sigma\rangle.
\]
\subsection*{Limit shape for random Young diagrams}
Let $\rho_{m}(\sigma)=\frac{d_{\sigma}^{2}}{m!}$ be the Plancherel measure. It was proven in \cite{L77} that for large $m$, a random Young diagram drawn from $\rho_m$ and rescaled by $\frac{1}{\sqrt{m}}$ in the $x$- and $y$-axes approaches a limiting shape. Here we imagine the Young diagram is drawn so that its two straight edges extend along the axes and meet at the origin. As $m\rightarrow\infty$, the resulting picture approaches a fixed curve with probability $\rightarrow1$, given by (\cite{R14}, Theorem 1.26) $\{(x(\theta),y(\theta):\,\theta\in[-\frac{\pi}{2},\frac{\pi}{2}]\}$
with 
\begin{align*}
x(\theta) & =\left(\frac{2\theta}{\pi}+1\right)\sin(\theta)+\frac{2}{\pi}\cos(\theta)\\
y(\theta) & =\left(\frac{2\theta}{\pi}-1\right)\sin(\theta)+\frac{2}{\pi}\cos(\theta).
\end{align*}
 This limiting curve is shown in Figure \ref{fig:limitshape}.

\begin{figure}
\centering
\includegraphics[scale=0.8]{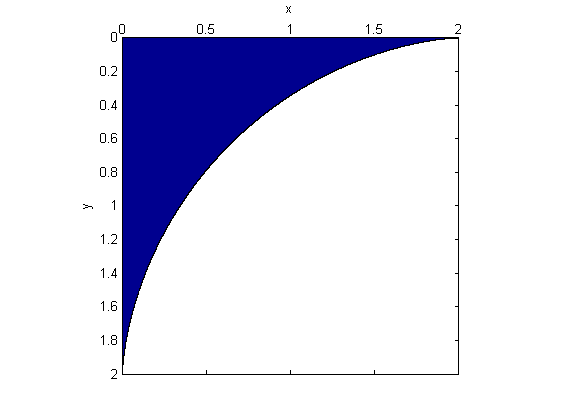}
\caption{A random Young diagram, drawn from the Plancherel measure and rescaled as discussed in the text, has the limiting shape shown in blue \cite{L77,R14}. \label{fig:limitshape}}
\end{figure}

\end{document}